\documentclass[12pt]{amsart}
\usepackage{latexsym,amssymb,amsmath,amsthm}
\usepackage{comment,graphicx,color,morefloats}
\usepackage[section]{placeins}
\newtheorem{thm}{Theorem}

\newtheorem{lemma}[thm]{Lemma}
\newtheorem{propo}[thm]{Proposition}
\theoremstyle{definition}
\newtheorem{defn}{Definition}

\def\B#1#2{{#1\choose #2}}

\title{The Jordan-Brouwer theorem for graphs}
\author{Oliver Knill}
\date{June 21, 2015}
\address{ Department of Mathematics \\ Harvard University \\ Cambridge, MA, 02138 }

\begin{document}
\maketitle

\begin{abstract}
We prove a discrete Jordan-Brouwer-Schoenflies separation theorem telling that a $(d-1)$-sphere 
$H$ embedded in a $d$-sphere $G$ defines two different connected graphs $A,B$ in $G$ such a way that 
$A \cap B = H$ and $A \cup B=G$ and such that the complementary graphs $A,B$ are both $d$-balls. 
The graph theoretic definitions are due to Evako: 
the unit sphere of a vertex $x$ of a graph $G=(V,E)$ is the graph generated by $\{ y \; | ; (x,y) \in E\}$. 
Inductively, a finite simple graph is called contractible if there is a vertex $x$ such that both 
its unit sphere $S(x)$ as well as the graph generated by $V \setminus \{x\}$ are contractible. 
Inductively, still following Evako, a $d$-sphere is a finite simple graph for which every unit 
sphere is a $(d-1)$-sphere and such that removing a single vertex renders the graph contractible. 
A $d$-ball $B$ is a contractible graph for which each unit sphere $S(x)$ 
is either a $(d-1)$-sphere in which case $x$ is called an interior point, or $S(x)$ 
is a $(d-1)$-ball in which case $x$ is called a boundary point and such that the set $\delta B$ of 
boundary point vertices generates a $(d-1)$-sphere. 
These inductive definitions are based on the assumption that the empty graph is the unique $(-1)$-sphere 
and that the one-point graph $K_1$ is the unique $0$-ball and that $K_1$ is contractible. 
The theorem needs the following notion of embedding:
a sphere $H$ is embedded in a graph $G$ if it is a subgraph of $G$ and if any intersection 
with any finite set of mutually neighboring unit spheres is a sphere. 
A knot of co-dimension $k$ in $G$ is a $(d-k)$-sphere $H$ embedded in a $d$-sphere $G$.
\end{abstract}

\section{Introduction}

The Jordan-Brouwer separation theorem \cite{Jordan,Brouwer12} assures that the image of an 
injective continuous map $H \to G$ from a $(d-1)$-sphere 
$H$ to a $d$-sphere $G$ divides $G$ into two compact connected regions $A,B$ such that 
$A \cup B = G$ and $A \cap B = H$. Under some regularity assumptions, the Schoenflies theorem assures 
that $A$ and $B$ are $d$-balls. 
Hypersphere embeddings belong to knot theory, the theory of embedding spheres in other spheres, and more generally
to manifold embedding theory \cite{DavermanVenema}. While $H$ is compact and homeomorphic to the standard sphere in ${\bf R}^{d}$, 
already a $1$-dimensional Jordan curve $H \subset {\bf R}^2$ can be complicated, as artwork in
\cite{RossRoss} or Osgood's construction of a Jordan curve of positive area
\cite{Osgood03} illustrate. The topology and regularity of the spheres as well as the dimension assumptions
matter: the result obviously does not hold surfaces $G$ of positive genus.
For codimension $2$ knots $H$ in a $3$-sphere $G$, the complement is connected but not simply connected.
Alexander \cite{Alexander24} gave the first example of a topological embedding of $S^2$ into $S^3$ 
for which one domain $A$ is simply connected while the other $B$ is not. With more regularity of $H$, 
the Mazur-Morse-Brown theorem \cite{Mazur59,Morse59,Brown62} assures that the complementary 
domains $A,B$ are homeomorphic to Euclidean unit balls 
if the embedding of $H$ is locally flat, a case which holds if $H$ is a smooth submanifold of $G$
diffeomorphic to a sphere. In the smooth case, all dimensions except $d=4$ are settled: 
one does not know whether there are smooth embeddings of $S^3$ into $S^4$ such that one of the domains is a
$4$-ball homeomorphic but not diffeomorphic to the Euclidean unit ball.
Related to this open Schoenflies problem is the open smooth Poincar\'e problem, which asks whether
there are is a smooth $4$-sphere homeomorphic but not diffeomorphic to the standard $4$-sphere. If the smooth 
Poincar\'e conjecture turns out to be true and no exotic smooth 4-spheres exist, then
also the Schoenflies conjecture would hold (a remark attributed in \cite{Calegari2013} to Friedman) 
as a Schoenflies counter example with an exotic $4$-ball would lead to an exotic $4$-sphere, a counter example
to smooth Poincar\'e.  \\

Even in the particular case of Jordan, various proof techniques are known. Jordan's
proof in \cite{Jordan} which was unjustly discredited at first \cite{Kline42} but rehabilitated in \cite{HalesJordanProof}. 
The Schoenflies theme is introduced in \cite{schoenflies1,schoenflies2,schoenflies3,schoenflies4}.
Brouwer \cite{Brouwer12} proves the higher dimensional theorem using $p$-dimensional ``nets" defined in Euclidean space.
His argument is similar to Jordan's proof for $d=2$ using an intersection number is what we will follow here. 
The theorem was used by Veblen \cite{Veblen05} to illustrate geometry he developed while writing his thesis
advised by Eliakim Moore. The Jordan curve case $d=2$ has become a test case for fully automated 
proof verifications. Its deepness in the case $d=2$ can be measured by the fact 
that "4000 instructions to the computer generate the proof of the Jordan curve theorem" 
\cite{Hales2007}. There are various proofs known of the Jordan-Brouwer theorem: it has 
been reduced to the Brouwer fixed point theorem \cite{Maehara}, proven using nonstandard analysis 
\cite{Narens} or dealt with using tools from complex analysis \cite{DostalTindell}. 
Alexander \cite{Alexander16} already used tools from algebraic topology and studied the cohomology of the
complementary domains when dealing with embeddings of with finite cellular chains. In some sense, we
follow here Alexander's take on the theorem, but in the language of graph theory, language formed
by A.V. Evako in \cite{I94,Evako1994} in the context of molecular spaces and digital topology.
It is also influenced by discrete Morse theory \cite{forman98,Forman1999}. 

\begin{figure}[ph]
\scalebox{0.14}{\includegraphics{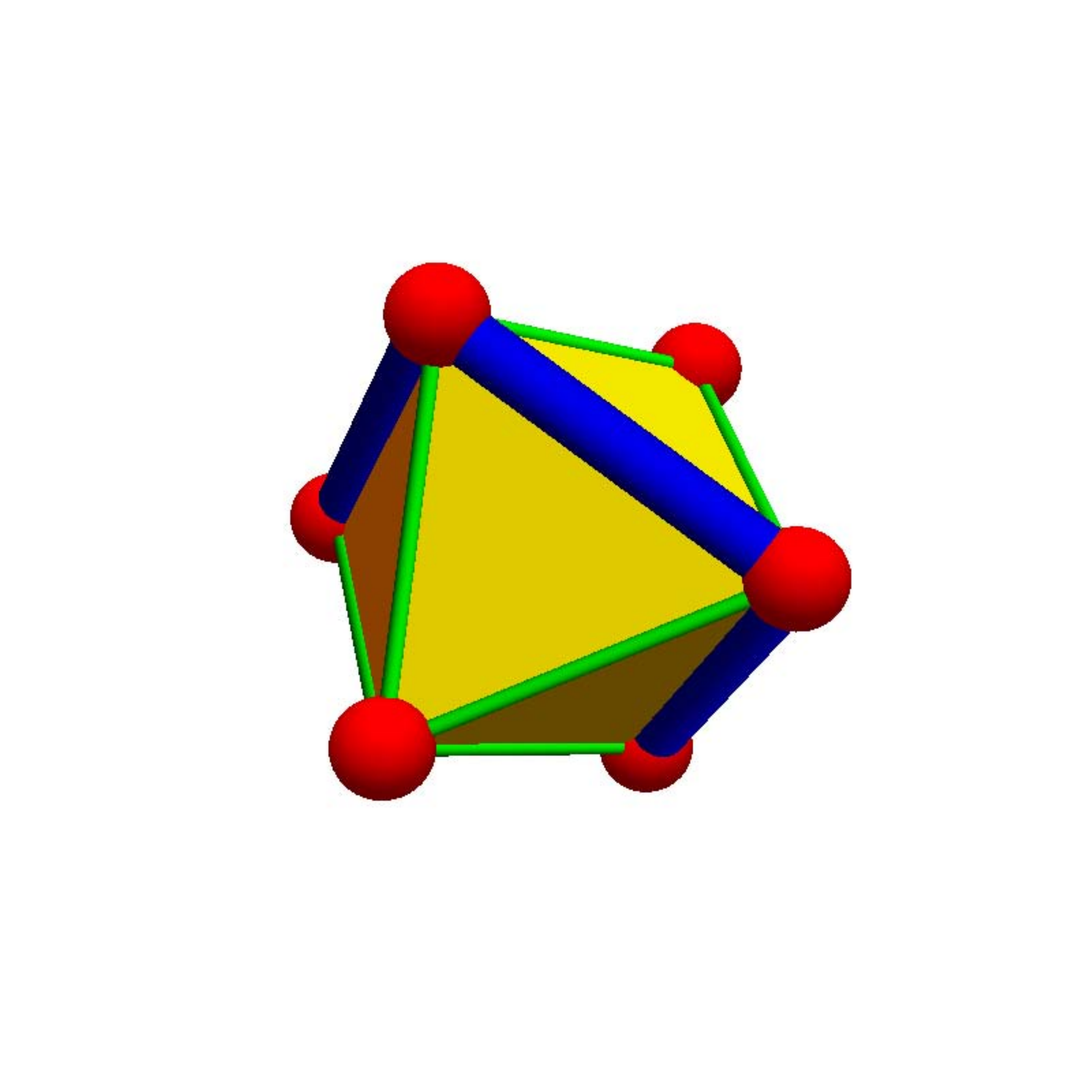}}
\scalebox{0.14}{\includegraphics{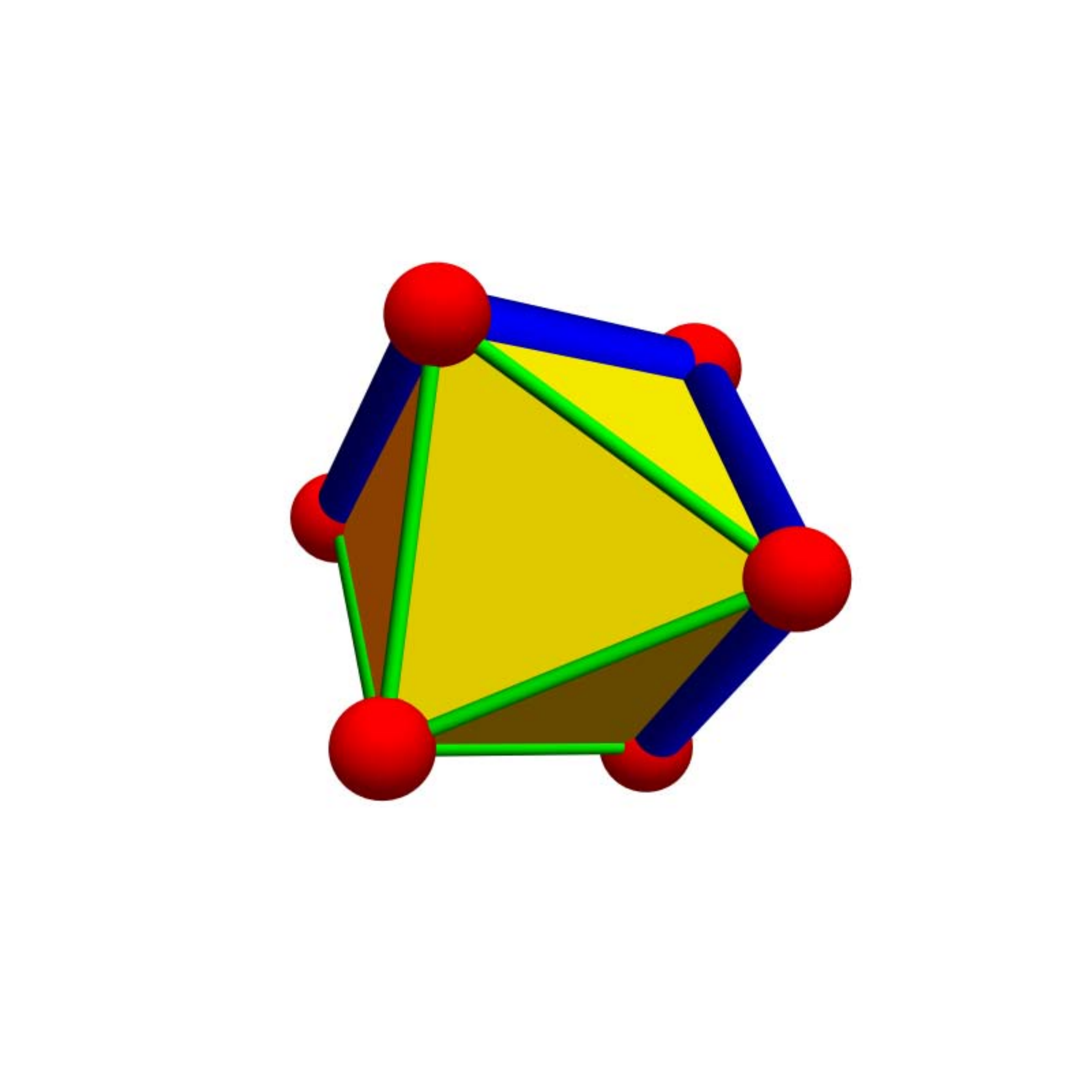}}
\scalebox{0.14}{\includegraphics{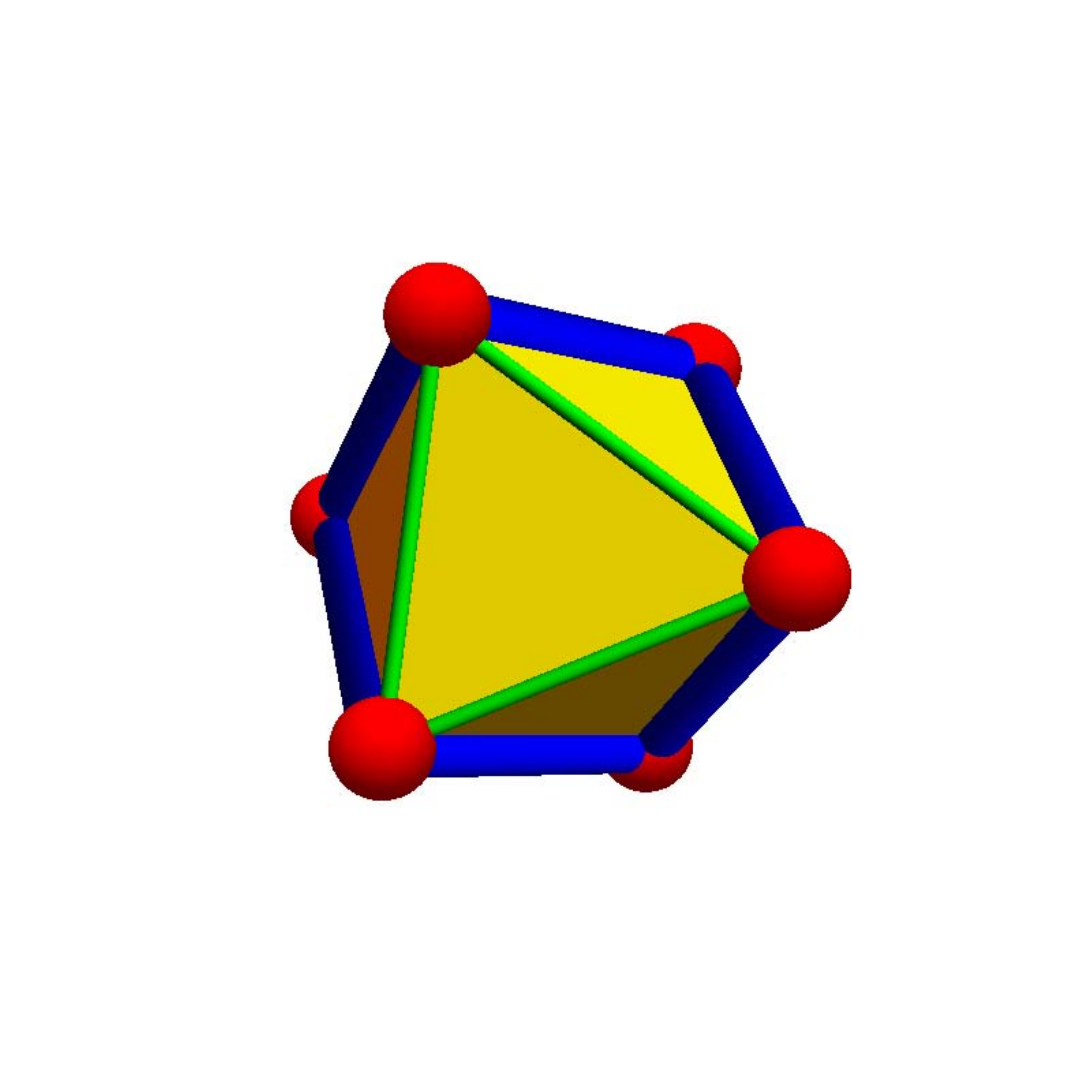}}
\scalebox{0.14}{\includegraphics{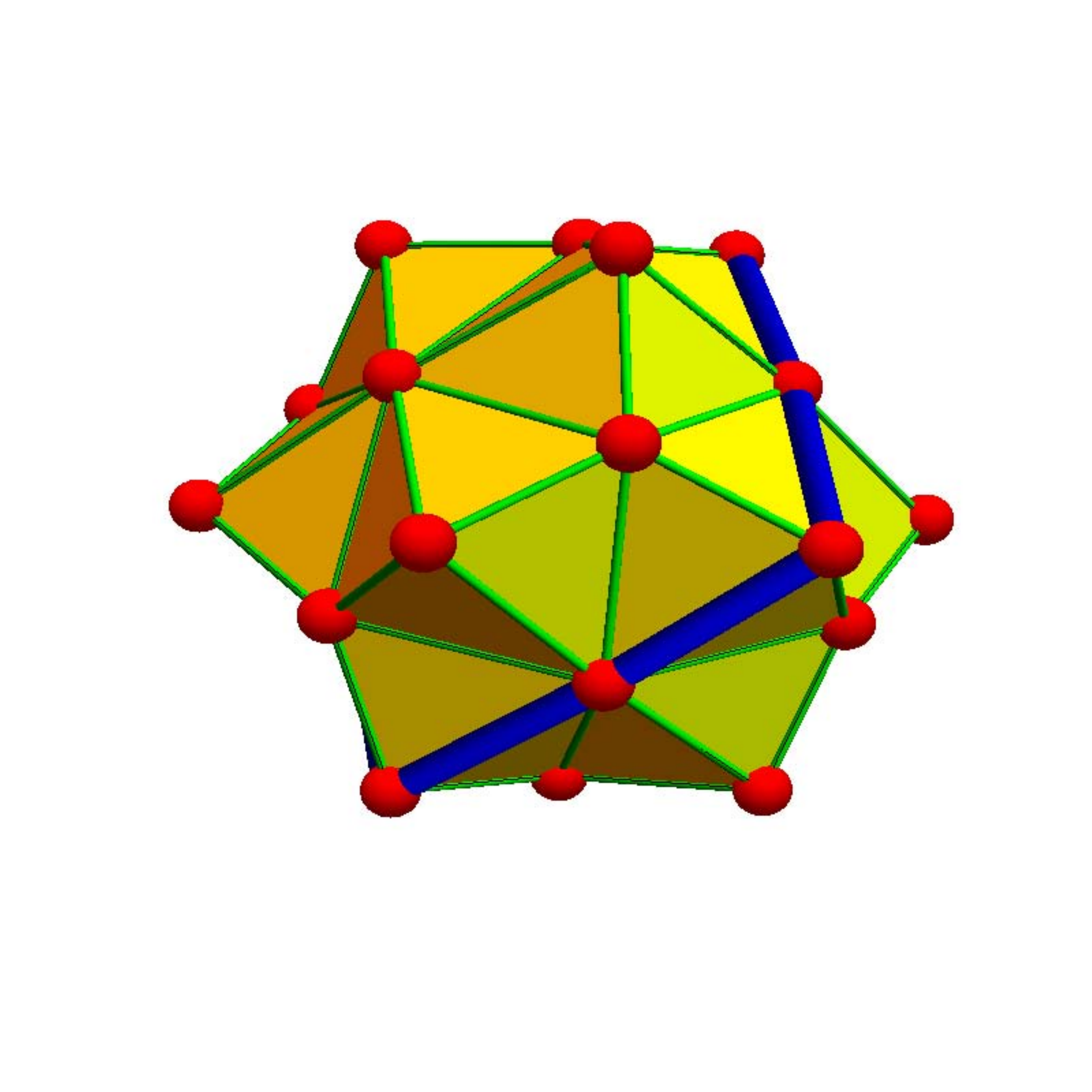}}
\scalebox{0.14}{\includegraphics{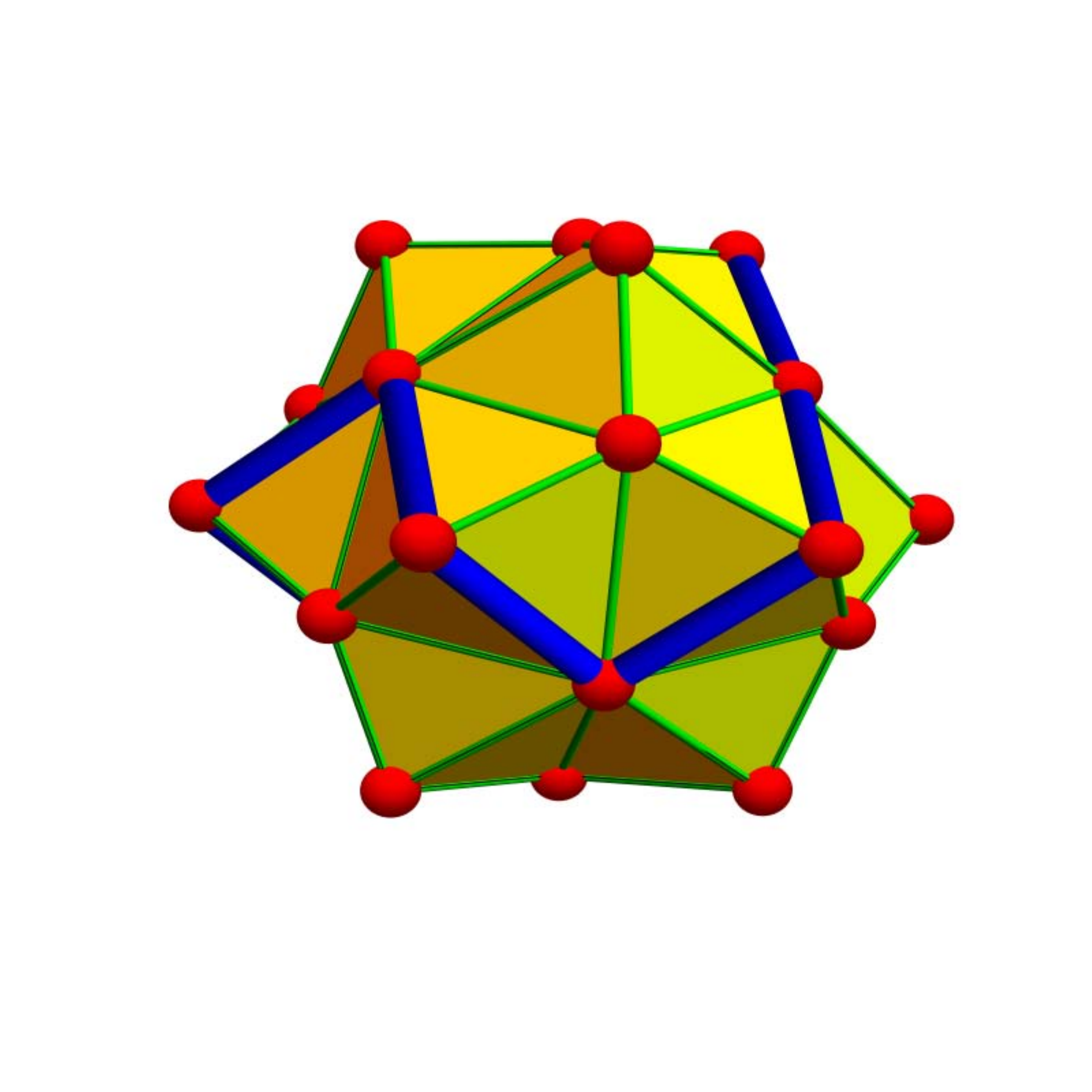}}
\scalebox{0.14}{\includegraphics{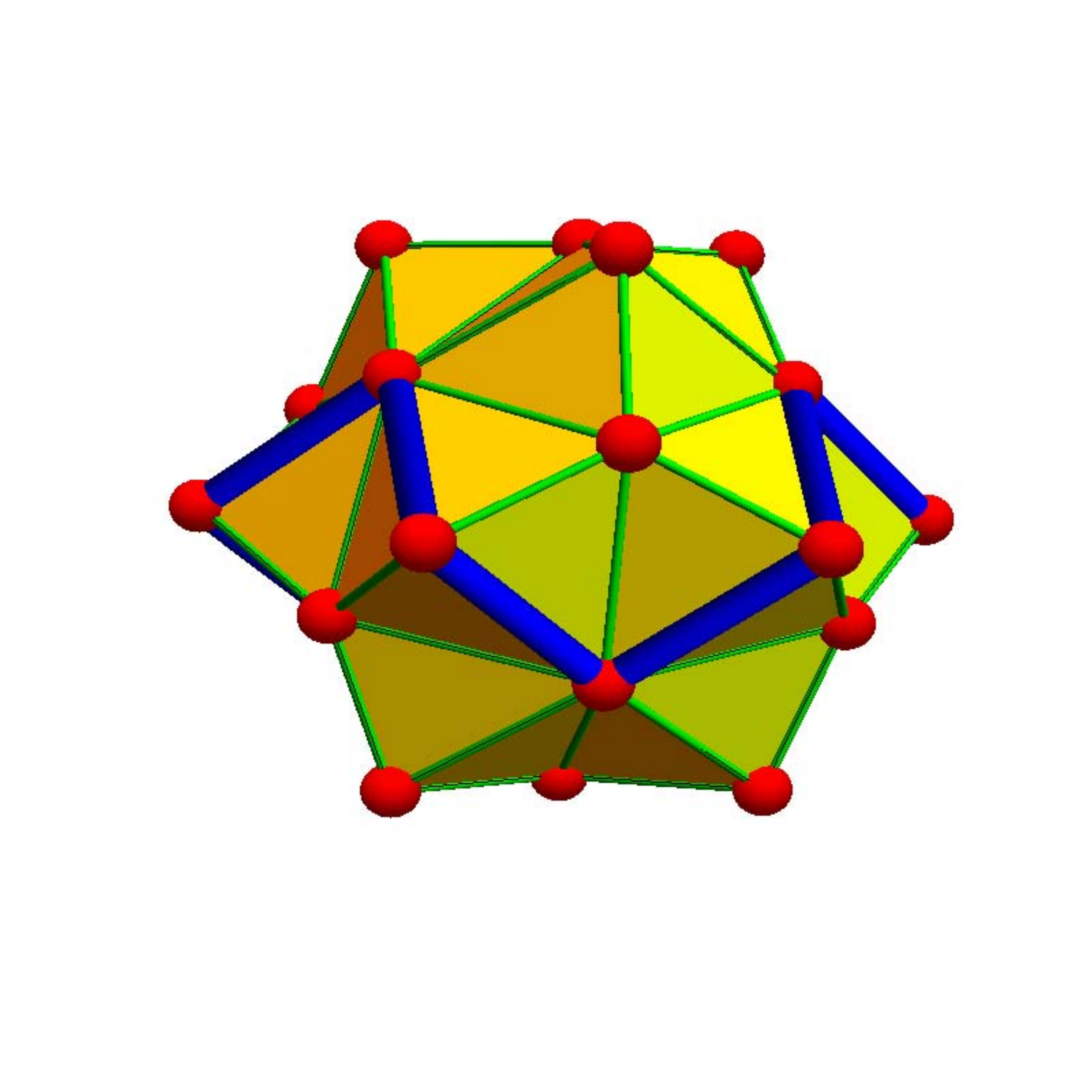}}
\caption{
\label{octahedron}
We see a homotopy deformation of an embedded $C_4$ graph $H_1$ in an 
octahedron $G$ to a Hamiltonian path $H_2$ of length $6$ in the $2$-sphere $G$.
The initial curve $H_1$ divides the octahedron into two complementary wheel graph domains $A,B$.
After two steps, the deformed $H_2$ is no more an embedding as its closure is $G$.
The deformations are obtained by taking a triangle $t$ containing an edge of $H$ and forming 
$H \to H \Delta t$. 
In the lower row, we see the situation on the simplex level, where the octahedron $G$ 
has become the Catalan solid $G_1$ and the deformed spheres remain sphere embeddings.
While our theorem will be formulated for
embedded spheres in $G$, the proof of the Schoenflies case requires the ability to
perform homotopy deformations and have a picture in which spheres remain embedded. 
The complementary domains $A,B$ in the lower case are $2$-balls: 
$9$ interior vertices representing $4$ triangles,
$4$ edges and the vertex of each pyramid at first.  After the deformation there are 
$7$ interior vertices representing the $4$ triangles and $3$ edges of the domain. 
On the simplex level the deformation steps are done using unit balls at vertices 
belonging to original triangles.
}
\end{figure}

When translating the theorem to the discrete, one has to specify what a ``sphere" and
what an ``embedding" of a sphere in an other sphere is in graph theory.
We also need notions of ``intersection numbers" of complementary spheres as well as
workable notions of ``homotopy deformations" of spheres within an other sphere.
Once the definitions are in place, the proof can be done by induction with respect to the dimension $d$.
Intersection numbers and the triviality of the fundamental group allow to show that the two components in the
$(d-1)$-dimensional unit sphere of a vertex in $G$ lifts to two components in the $d$-dimensional case:
to prove that there are two complementary components one has to verify that the intersection number of a 
closed curve with the
$(d-1)$-sphere $H$ is even. This implies that if a curve from a point in $A$ to $B$ in the smaller
dimensional case with intersection number $1$ is complemented to become a closed curve in $G$, also the new
connection from $A$ to $B$ has an odd intersection number, preventing the two regions to be the same.
The discrete notions are close to the ``parity functions" used by Jordan explained in \cite{HalesJordanProof}.
Having established that the complement of $H$ has exactly two components $A,B$,
a homotopy deformation of $H$ to a simplex in $A$ or to a simplex $B$ will establish the Schoenflies statement
that $A,B$ are $d$-balls. We will have to describe the homotopy on the regularized simplex level to
regularizes things. The Jordan-Schoenflies theme is here used as a test bed
for definitions in graph theory. Indeed, we rely on ideas from \cite{KnillTopology,KnillProduct}. \\

Lets look at the Jordan case, an embedding of a circle in a 2-sphere. A naive version of a discrete
Jordan theorem is the statement that a ``simple closed curve in a discrete sphere divides the complement into two regions".
As illustrated in Figure~(\ref{octahedron}), this only holds with a grain of salt.
Take the octahedron graph $G$ and a closed Hamiltonian path which visits all vertices, leaving no
complement. It is even possible for any $m \geq 0$ to construct a discrete sphere $G$ and a curve $H$
in $G$ such that the complement has $m$ components as the curve can bubble off regions by coming close to itself without
intersecting itself but still dividing up a disk from the rest of the sphere.
As usual with failures of discrete versions of continuum results, the culprit is the definition, in this case, it is
the definition of an ``embedding". The example of a Hamiltonian path is more like a discrete Peano
curve in the continuum as it visits all vertices but without hitting all directions or area forms of the plane.
We need to make sure that also the closure of the embedding is the same curve. In the graph theoretical
case, we ask that the graph generated in $G$ by the vertex set of $H$ is still a sphere.
For curves in a 2-sphere for example we have to ask that the embedded curve features no triples of vertices forming a
triangle in $G$. There is an other reinterpretation to 
make the theorem true for simple closed curves as we will see in the proof: there is a regularized picture on the
simplex level, where two complementary domains always exist.
We need the simplex regularization because embeddings do not play well with deformations: When making
homotopy deformation steps, we in general lose the property of having an embedded sphere.
Already in discrete planar geometry, where we work in the flat $2$-dimensional hexagonal lattice
\cite{elemente11} one has to invoke rather subtle definitions to get conditions which make things work.
Discrete topological properties very much depend on the definitions used. 
It would be possible to build a homotopy deformation process which honors the embeddability, but it could be 
complicated. The construction of a graph product \cite{KnillProduct} provided us with an elegant resolution
of the problem: we can watch the deformation of an ``enhanced embedding" $H_1$ in $G_1$, where $H_1$ is the graph
obtained from $H$ by taking all the complete subgraphs of $H$ as vertices and connecting two of them if
one is a subgraph of the other.
It turns out that even if $H$ is only a subgraph of $G$, the graph $H_1$ is an embedding of $G_1$.
This holds in particular if $H$ is a Hamiltonian path in $G$, the closure of $H_1$, the graph
generated by the vertex set of $H_1$ remains geometric in $G_1$.
The Cartesian product \cite{KnillProduct} allows also to look at homotopy groups geometrically:
a deformation of a curve in a $2$-sphere $G$ for example
is now described as a geometric surface in the $3$-dimensional solid cylinder $G \times L_n$, where
$L_n$ is the $1$-dimensional line graph with $n$ vertices and $(n-1)$ edges. 
This is now close to the definition of homotopy of a curve using a function $F(t,s)$ in two variables 
so that $F(t,0)$ is the first curve and $F(t,1)$ the second.  \\

This paper is not the first take on a discrete Jordan theorem as
various translations of the Jordan theorem have been constructed to the discrete. They are
all different from what we do here: \cite{LiChen} uses notions of discrete geometry,
\cite{HermanDigitalSpaces} looks at Jordan surfaces in the geometry of digital spaces, 
\cite{Evako2013} proves a Jordan-Brouwer result in the discrete lattice $Z^d$,
\cite{Surowka} extends a result of Steinhaus on a $m \times n$
checkerboard $G$, a minimal king-path $H$ connecting two not rook-adjacent elements of
the boundary divides $G$ into two components. A variant of this on a hexagonal board
\cite{Gale} uses such a result to prove the Brouwer fixed point theorem in two dimensions.
\cite{VinceLittle} deals more generally with graphs which can have multiple connections.
The result essentially establishes what we do in the special case $d=2$.
\cite{Stout} looks at two theorems in $L \times L$ or $L \star L$, where $\times$
is the usual Cartesian product and $\star$ the tight product. In both cases,
the complement of $C$ has exactly two path components. The paper \cite{NeumannLaraWilson}
deals with graphs, for which every unit sphere $S(x)$ is a Hamiltonian graph. Also here,
a closed simple path $C$ in a connected planar graph $G$ divides
the complement into exactly two components. \cite{VinceLittle} deals with graphs which can have
multiple connections. On $2$-spheres, it resembles the Jordan case $d=2$ covered here
as a special case.  \\

The main theorems given here could readily be derived from the continuum by building
a smooth manifold from a graph, and then use the Jordan-Brouwer-Schoenflies rsp. Mazur-Morse-Brown 
theorem. The approach however is different as no continuum is involved: all definitions and steps
are combinatorial and self-contained and could be accepted by a mathematician avoiding  axioms
invoking infinity. Sometimes, constructabily can be a goal \cite{ConstructiveJordan}. 
New is that we can prove Jordan-Brouwer-Schoenflies entirely within graph
theory using a general inductive graph theoretical notion of ``sphere" \cite{Evako1994}. Papers
of Jordan, Brouwer or Alexander show that the proofs in the continuum often 
deal with a combinatorial part only and then use an approximation argument to get the general case. 
As the Alexander horned sphere, the open Schoenfliess conjecture or questions about triangulations related 
to the Hauptvermutung show, the approximation part can be difficult within topology and we don't go into it.
Our proof remains discrete but essentially follows the arguments from the continuum by 
defining intersection numbers and use induction with respect to dimension. The induction proof is 
possible because of the recursive definition of spheres and seems not
have been used in the continuum, nor in discrete geometry or graph theory. 
But what really makes the theorem go, is to watch the story on the simplex level, where 
geometric graphs $G$ remain geometric and where sub-spheres $H$ of $G$ can be watched as 
embedded spheres $H_1$ in $G_1$. The step $G \to G_1$ can be used in the theory of triangulations 
because it has a ``regularizing effect". Given a triangulation $G$ described as a graph, the 
new triangulation $G_1$ has nicer properties like that the unit sphere of a point is now a 
graph theoretically defined sphere. 

\section{Definitions}

\begin{defn}
A subset $W$ of the vertex set $V=V(G)$ of a graph $G$ {\bf generates} a subgraph $(W,F)$ of $G$,
where $F$ is defined as the subset $\{ (a,b) \in E \; | \; a \in W, b \in W \; \}$ of the edge set of $E$. 
The {\bf unit sphere} $S(x)$ of a vertex $x$ in a graph is the subgraph of $G$ generated by all 
neighboring vertices of $x$. The {\bf unit ball} $B(x)$ is the subgraph of $G$ generated by the 
union of $\{x\}$ and the vertex set of the unit sphere $S(x)$. 
\end{defn}

If $H$ is a subgraph of $G$, one can think of the graph generated by $H$ within $G$ as a ``closure"
of $H$ within $G$. It is in general larger than $H$. For example, the closure of the line 
graph $H=( (a,b,c), (a,b),b,c))$ within the complete graph $G$ with vertex set $V=\{ a,b,c \; \}$ 
is equal to $G$. 

\begin{defn}
Starting with the assumption that the one point graph $K_1$ is contractible,
recursively define a finite simple graph $G=(G,V)$ to be {\bf contractible} if it 
contains a vertex $x$ such that its unit sphere $S(x)$
as well as the graph generated by $V \setminus \{x\}$ are both contractible. 
\end{defn}

\begin{defn}
A complete subgraph $K_{k+1}$ of $G$ will also be denoted {\bf $k$-dimensional simplex}.
If $v_k(G)$ is the set of $k$-dimensional simplices in $G$,
then the {\bf Euler characteristic} of $G$ is defined as $\chi(G) = \sum_{k=0} (-1)^k v_k(G)$. 
\end{defn}

{\bf Examples}. \\
{\bf 1)} The Euler characteristic of a contractible graph is always $1$ as removing one vertex
does not change it. One can use that $\chi(A \cap B) = \chi(A) + \chi(B) - \chi(A \cap B)$
and use inductively the assumption that unit balls as well as the spheres $S(x)$ in a reduction are
both contractible. \\
{\bf 2)} Also by induction, using that a unit sphere of a $d$-sphere is $(d-1)$-sphere, one verified
that $\chi(G) = 1+(-1)^d$ for a $d$-sphere. This holds also in the case $d=-1$ as the 
Euler characteristic of the empty graph is $0$. The Euler characteristic of an octahedron for
example is $6-12+8=2$ as there are $6$ vertices, 12 edges and 8 triangles. 
The cube graph $G$ is not a sphere as the unit sphere at each vertex is $P_3$. Its Euler characteristic 
is $8-12=-4$. $G$ is a sphere with $6$ holes punched in, leaving only a $1$-dimensional
skeleton. A $2$-dimensional cube can be constructed as the boundary $\delta B$ 
of the solid cube $B=L_2 \times L_2 \times L_2$ defined in \cite{KnillProduct}.

\begin{defn}
Removing a vertex $x$ from $G$ for which $S(x)$ is contractible is called a 
{\bf homotopy reduction step}. The inverse operation of performing a suspension over
a contractible subgraph $H$ of $G$ by adding a new vertex $x$ and connecting $x$
to all the vertices in $H$ is called a {\bf homotopy extension step}. 
A finite composition of reduction or extension steps is called a 
{\bf homotopy deformation} of the graph $G$.
\end{defn}

{\bf Remarks.}\\
{\bf 1)} Since simple homotopy steps removing or adding vertices with contractible $S(x)$
do not change the Euler characteristic, it is a function on the homotopy 
classes \cite{I94}. If we add a vertex for which $S(x)$ is not contractible, we add a vertex with 
index $1-\chi(S(x))$ which is a Poincar\'e-Hopf index \cite{poincarehopf}. Given a function 
on the vertex set giving an ordering on the build up of the graph one gets the Poincar\'e-Hopf theorem. \\
{\bf 2)} Examples like the Bing house or the Dunce hat show that
homotopic to a one-point graph $K_1$ is not equivalent to contractible: some graphs 
might have to be expanded first before being contractible. This is relevant in 
Lusternik-Schnirelmann category \cite{josellisknill}. \\
{\bf 3)} The discrete notion of homotopy builds an equivalence relation on graphs in 
the same way homotopy does in the continuum. The problem of classifying homotopy types can not
be refined as one can ask how many types there are on graphs with $n$ vertices. \\
{\bf 4)} The discrete deformation steps were
put forward by Whitehead \cite{Whitehead} in the context of cellular complexes. 
The graph version is due to \cite{I94} and was simplified in \cite{CYY}. \\
{\bf 5)} The definition of $d$-spheres and $d$-balls is inductive. Introduced in
\cite{knillgraphcoloring,knillgraphcoloring2} we were puzzled then why this natural setup has 
not appeared before. But it actually has, in the context of digital topology \cite{Evako2013}
going back to \cite{Evako1994} and we should call such spheres Evako spheres. 
Alexander Evako is a name shortcut for Alexander Ivashchenko who also introduced homotopy 
to graph theory and also as I only learned now while reviewing his work found in \cite{I94a}
a similar higher dimensional Gauss-Bonnet-Chern theorem \cite{cherngaussbonnet} in graph theory.

\begin{defn}
The induction starts with the assumption that the empty graph is the only $(-1)$-sphere 
and that the graph $K_1$ is the only $0$-ball. 
A graph is called a {\bf d-sphere}, if all its unit spheres $S(x)$ are $(d-1)$ spheres and 
if there exists a vertex $x \in V$ such that the graph generated by $V \setminus \{x\}$ 
is contractible. A contractible graph $G$ is called a {\bf $d$-ball}, if one can partition its
vertex set $V$ into two sets ${\rm int}(V) = \{ x \in V \; | \; S(x)$ is a $(d-1)$ sphere $\}$ and 
$\delta(V) = \{ x \in V \; | \; S(x)$ is a $(d-1)$-ball $\}$ such 
that $\delta V$ generates a $(d-1)$-sphere called $\delta G$, the {\bf boundary of $G$}.
\end{defn}

By induction, $\chi(G)= 1+(-1)^d$  if $G$ is a $d$-sphere and $\chi(G)=1$ if $G$ is a $d$-ball. \\

{\bf Examples.} \\
{\bf 1)} The boundary sphere of the $0$-ball $K_1$ is the $(-1)$ sphere $\emptyset$, the empty graph.  \\
{\bf 2)} The boundary sphere of the line graph $L_n$ with $ n>2$ vertices, is the $0$-sphere $P_2$. 
Line graphs are $1$-balls. \\
{\bf 3)} The boundary sphere of the wheel graph $W_n$ with $n \geq 4$ is the circular graph $C_n$. 
The wheel graph is an example of a $2$-ball and $C_n$ is an example a $1$-sphere.  \\
{\bf 4)} The boundary sphere of the $3$-ball obtained by making 
a suspension of a point with the octahedron is the octahedron itself. \\
{\bf 5)} We defined in \cite{knillgraphcoloring2} Platonic spheres as $d$-spheres for which all 
unit spheres are Platonic $(d-1)$-spheres. This definition has been given already by Evako. 
The discrete Gauss-Bonnet-Chern theorem \cite{cherngaussbonnet}
easily allows a classification: all $1$-dimensional spheres $C_n, n>3$ are
Platonic for $d=1$, the Octahedron and Icosahedron are the two Platonic $2$-spheres, 
the sixteen and six-hundred cells are the Platonic $3$-spheres. As we only now realize while looking over
the work of Evako, we noticed that the Gauss-Bonnet theorem \cite{cherngaussbonnet} appears in 
\cite{I94a}. The $d$-cross polytop $P_2 \star P_2 \star \cdots \star P_2$ obtained by repeating 
suspension operations from the $0$-sphere $P_2$ is the unique Platonic $d$-sphere for $d>3$.  

\begin{defn}
The dimension of a graph is inductively defined as ${\rm dim}(G)=1+\sum_{x \in V} {\rm dim}(S(x))/v_0$, where
$v_0=|V|$ is the cardinality of the vertex set. The induction foundation is that the empty graph $\emptyset$
has dimension $0$. The dimension of a finite simple graph is a rational number. 
\end{defn}

{\bf Remarks.} \\
{\bf 1)} This inductive dimension for graphs has appeared first in \cite{elemente11,randomgraph}. It is motivated by 
the Menger-Uryson dimension in the continuum but it is different because with respect to the metric on a graph, the
Menger-Uryson dimension is $0$. \\
{\bf 2)} Much of graph theory literature ignores the Whitney simplex 
structure and treat graphs as one dimensional simplicial complexes. The inductive dimension behaves very much 
like the Hausdorff dimension in the continuum, the product \cite{KnillProduct} is super additive
${\rm dim}(H \times K) \geq {\rm dim}(H) + {\rm dim}(K)$ like Hausdorff
dimension of sets in Euclidean space. \\
{\bf 3)} There are related notions of dimension like \cite{SmythTsaurStewart}, who look at the largest dimension
of a complete graph and then extend the dimension using the usual Cartesian product. This is not equivalent
to the dimension given above. \\

{\bf Examples.} \\
{\bf 1)} The complete graph $K_n$ has dimension $n-1$.
{\bf 2)} The dimension of the house graph obtained by gluing $C_4$ to $K_3$ along an edge is
$22/15$: there are two unit spheres of dimension $0$ which are the base points,
two unit spheres of dimension $2/3$ corresponding to the two lower roof points and 
one unit sphere of dimension $1$ which is the tip of the roof. \\
{\bf 3)} The expectation $d_n(p)$ of dimension 
on Erdoes-Renyi probability spaces $G(n,p)$ of all subgraphs in $K_n$ for which
edges are turned on with probability $p$ can be computed explicitly. It is an explicit
polynomial in $p$ given by $d_{n+1}(p) = 1+\sum_{k=0}^n \B{n}{k} p^k (1-p)^{n-k} d_k(p)$ 
\cite{randomgraph}.

\begin{defn}
A finite simple graph $G=(V,E)$ is called a {\bf geometric graph} of dimension $d$ if every unit sphere $S(x)$ is 
a $(d-1)$-sphere. A finite simple graph $G$ is a {\bf geometric graph with boundary}
if every unit sphere $S(x)$ is either a $(d-1)$ sphere or a $(d-1)$-ball. The subset of the vertex set $V$,
in which $S(x)$ is a ball generates the {\bf boundary graph} of $G$. We denote it by $\delta G$ and assume it
to be geometric of dimension $(d-1)$. 
\end{defn} 

{\bf Examples.} \\
{\bf 1)} By definition, $d$-balls are geometric graphs of dimension $d$ and $d$-spheres are 
geometric graphs of dimension $d$. \\
{\bf 2)} For every smooth $d$-manifold one can look at triangulations which are geometric 
$d$-graphs. The class of triangulations is much larger. \\

{\bf Remarks.} \\
{\bf 1)} Geometric graphs play the role of manifolds. By embedding each discrete unit ball $B(x)$
in an Euclidean space and patching these charts together one can from every geometric graph $G$ generate
a smooth compact manifold $M$. Similarly, if $G$ is a geometric graph with boundary, one can ``fill it up"
to generate from it a compact manifold with boundary. \\
{\bf 2)} The just mentioned obvious functor from geometric graphs to manifolds 
is analogue to the construction of manifolds from simplicial complexes. We don't want to 
use this functor for proofs and remain in the category of graphs. One reason is that many 
computer algebra systems have the category of graphs built in as a fundamental data structure. An other reason
is that we want to explore notions in graph theory and stay combinatorial. \\
{\bf 3)} Graph theory avoids also the rather difficult notion of triangularization. Many triangularizations
are not geometric. In topology, one would for example consider the tetrahedron graph $K_4$ as a triangulation
of the $2$-sphere. But $K_4$ is not a sphere because unit spheres are $K_3$ which are not spheres etc.
And $K_4$ is also not a ball. While it is contractible, it coincides with its boundary as it does not have 
interior points. Free after Euclid one could say that $K_{d+1}$ is a {\bf $d$-dimensional point}, as 
it has no $d$-dimensional parts.  \\
{\bf 4)} Every graph defines a simplicial complex, which is sometimes called
the Whitney complex, but graphs are a different category than simplicial complexes. 
Algebraically, $x+y+z+xy+yz+xz$ is a simplicial complex which is not a graph as it does 
not contain the triangle simplex. 
The graph completion $K_3$ described by $x+y+z+xy+yz+zx+xyz$ however, is a graph. 

\begin{defn}
A geometric graph of dimension $d$ is called {\bf orientable} if one can assign a
permutation to each of its $d$-dimensional simplices in such a way that one has 
compatibility of the induced permutations on the intersections of neighboring simplices.
For an orientable graph, there is a constant non-zero $d$-form $f$, called {\bf volume form}.
It satisfies $df=0$ for the exterior derivative $d$ but which can not be written as $dh$.
\end{defn}

{\bf Remarks.} \\
{\bf 1)} A connected orientable $d$-dimensional
geometric graph has a $1$-dimensional cohomology group $H^d(G)$.
This is a special case of Poincar\'e duality, assuring an isomorphism of 
$H^{n-d}(G)$ with $H^n(G)$, which holds for all geometric graphs. \\
{\bf 2)} For geometric graphs, an orientation induces an orientation on the boundary. 
Stokes theorem for geometric graphs with boundary is $\int_G df = \int_{\delta G} f$ 
\cite{knillcalculus}, as it is the definition on each simplex. \\

{\bf Examples.} \\
{\bf 1)} All $d$-spheres with $d \geq 1$ are examples of orientable graphs. \\
{\bf 2)} If a $d$-sphere $G$ has the property that antipodal points have distance at least 4, then 
the antipodal identification map $T$ factors out a geometric graph $G/T$, we get a discrete projective 
space $P^d$. For even dimensions $d$, this geometric graph is not orientable. \\
{\bf 3)} The cylinder $C_n \times L_m$, with $n \geq 4, m \geq 2$ is orientable. One can get a sphere,
a projective plane, a Klein bottle or a torus from identifications of the boundary of $L_n \times L_m$ 
in the same way as in the continuum. For example, the graph $L_3 \times L_3$ is obtained by taking the 
25 polynomial monoid entries of $(a+ab+b+bc+c) (u+uv+v+vw+w)$ as vertices and connecting 
two if one divides the other. \\

The following definition of the graph product has been given in \cite{KnillProduct}:

\begin{defn}
A graph $G=(V,E)$ with vertex set $V=\{x_1,\dots,x_n \; \}$ defines a polynomial
$f_G(x_1,\dots,x_n) = \sum_{x} x$, where $x = x_1^{k_1} \cdots x_n^{k_n}$
with $k_i \in \{0,1\}$ represents a complete subgraph of $G$.
The polynomial defines the graph $G_1$ with vertex set $V=\{ x \; | \; {\rm simplex} \; \}$ and
edge set $E=\{ (x,y) \; | \; x|y \; {\rm or} \; y|x \}$,
where $x|y$ means $x$ divides $y$, geometrically meaning that $x$ is a sub-simplex of $y$.
Given two graphs $H,K$, define its graph product $H \times K = G(f_H \cdot f_K)$.
The graph $G_1 = G \times K_1$ is called the {\bf enhanced graph} obtained from $G$.
\end{defn}

{\bf Examples.} \\
{\bf 1)} For $G=C_4$ we have $f_G=x+xy+y+yz+z+zw+w+wx$ and $G_1=C_8$. \\
{\bf 2)} For $G=K_3$, we have $f_G=x+y+z+xy+yz+zx+xyz$ and $G_1=W_6$. \\
{\bf 3)} For a graph without triangles, $G_1$ is homeomorphic to $G$ in the classical sense. \\

{\bf Remarks.}\\
{\bf 1)} The graph $G_1$ has as vertices the complete subgraphs of $G$. Two simplices are connected if 
and only one is contained in the other. If $G$ is geometric, then $G_1$ is geometric. For example, if 
$G$ is the octahedron graph with $v_0=6$ vertices, $v_1=12$ edges and $v_2=8$ triangles,
then $G_1$ is the graph belonging to the Catalan solid with $v_0+v_1+v_2=26$ vertices 
and which has triangular faces. Also $G_1$ is a $2$-sphere. \\
{\bf 2)} In full generality, the graph $G_1$ is homotopic to $G$ and has therefore
the same cohomology. The unit balls of $G_1$ form a weak \v{C}ech cover in the sense that the nerve 
graph of the cover is the old graph $G$ and two elements in the cover are linked, 
if their intersection is a $d-1$ dimensional graph. To get from $G_1$ to $G$, successively shrink
each unit ball of original vertices analogue to a Vietoris-Begle theorem.
If $G$ is geometric, it is possible to modify the cover to have it homeomorphic in the sense of \cite{KnillTopology} 
so that if $G$ is geometric then $G$ and $G_1$ are homeomorphic. 
As the dimension of $G_1$ can be slightly larger in general, the property that $G_1$ and $G$ are
homeomorphic for all general finite simple graphs does not hold. \\

{\bf Examples.} \\
{\bf 1)} A Hamiltonian path $H$ in the icosahedron $G$, a $2$-sphere,
does not leave any room for complementary domains. However, the graph $H_1$ in $G_1$ divides
$G_1$ into two regions. The graph $G_1$ by the way is the disdyakis triacontahedron, a 
Catalan solid with $62$ vertices. \\
{\bf 2)} Let $G$ be the octahedron, a $2$-sphere with $6$ vertices. 
Assume $a,b,c,d$ are the vertices of the equator sphere
$C_4$ and that $n,p$ are the north and south pole. Define the finite simple curve 
$a,b,c,d,s,a$ of length $5$. It is a circle $H$ in $G$, but it is not embedded. 
In this case, the complement of $H$ has only one region $A=\{n\}$. The Jordan-Brouwer theorem
is false. But its only in this picture. If we look at the 
embedding of $H_1$ in $G_1$, then this is an embedding which divides 
the Catalan solid $G_1$ into two regions $A_1,B_1$. 

\begin{defn}
A graph $H=(W,F)$ is a subgraph of $G=(V,E)$ if $W \subset W$ and $F \subset E$. 
Let $H,G$ be geometric graphs. 
A graph $H$ is {\bf embedded} in an other graph $G$ if $H$ is a subgraph such that
for any collection of unit spheres $S(x_j)$ in $G$, where $x_j$ is in the vertex set of some $K_k$,
the intersection $H \cap \bigcap_{j=1}^k S(x_j)$ is a sphere. 
\end{defn}

{\bf Examples.} \\
{\bf 1)} A $0$-sphere $H$ is embedded in a geometric graph if the two vertices
are not adjacent. A $1$-sphere $H$ is embedded if two vertices of $H$ are connected in $G$
if and only if they are connected in $H$. \\
{\bf 2)} A graph $C_k$ can only be immersed naturally in $C_n$ if $n$ divides $k$ and $n \geq 4$.
It is a curve winding $k/n$ times around $C_n$. 
For example, $C_{15}$ can be immersed in $C_5$ and described by the homomorphism
alebraically described by $x \in Z_{15} \to x \in Z_5$ if $C_n$ is identified with 
the additive group $Z_n$. In this algebraic setting dealing with the fundamental group
it is better to look at the graph homomorphism rather than the physical image of the 
homomorphism. \\
{\bf 3)} If $H$ is embedded in $G$, then also $H_1$ is an embedding of $G_1$. 
But $H_1$ is an embedding in $G_1$ even if $H$ is only a subgraph of $G$. 
See Figure~(\ref{octahedron}). \\

The following definition places the sphere 
embedding problem into the larger context of knot theory: 

\begin{defn}
A {\bf knot of co-dimension $k$} is an embedding of a $(d-k)$-sphere $H$ 
in a $d$-sphere $G$. 
\end{defn}

{\bf Remarks.} \\
{\bf 1)} A knot can be called {\bf trivial} if it is homeomorphic to the
$(d-k)$-cross polytop embedded in the $d$-cross polytop in the sense of 
\cite{KnillTopology}.  \\
{\bf 2)} As we don't yet know whether there are graphs homeomorphic to spheres which are not spheres
or whether there are graphs homeomorphic to balls which are not balls.  \\
{\bf 3)} The Jordan-Brouwer-Schoenflies
theorem can not be stated in the form that a $(d-1)$-sphere in a $d$ sphere is trivial. 

\begin{defn}
A {\bf closed curve} in a graph is a sequence of vertices $x_j$ with $(x_j,x_{j+1}) \in E$
and $x_n=x_0$. A {\bf simple curve} in the graph is the image of an injective graph homomorphism $L_n \to G$,
where $L_n$ is the line graph. 
A {\bf simple closed curve} is the image of an injective homomorphism $C_n \to G$ with $n \geq 3$. 
It is an embedding of the circle if the image generates a circle. In general, a simple closed curve is not
an embedding of a circular graph.
\end{defn}

{\bf Example.} \\
{\bf 1)} A Hamiltonian path is a simple closed curve in a graph $G$ which visits all vertices
exactly once. Such a path is not an embedding if $G$ has dimension larger than $1$
as illustrated in Figure~(\ref{octahedron}). \\
{\bf 2)} While we mainly deal with geometric graphs, graphs for which all unit spheres are 
spheres, the notion of a simple closed curve or an embedding can be generalized for any pair of 
finite simple graphs $H,G$: if $H$ is a subgraph of $G$, then there is an injective graph
homomorphism from $H$ to $G$. If the intersection of an intersection of finitely many neighboring unit spheres
with $H$ is a sphere, we speak of an embedding.  

\begin{defn}
An embedding of a graph $H$ in $G$ {\bf separates} $G$ into two graphs $A,B$ if 
$A \cap B = H, A \cup B = G$ and $A \setminus H, B \setminus H$ 
are disjoint nonempty graphs. The two graphs $A,B$ are called 
{\bf complementary subgraphs} of the embedding $H$ in $G$. 
\end{defn}

{\bf Examples.} \\
{\bf 1)} The empty graph separates any two connectivity components of a graph. \\
{\bf 2)} By definition, if a graph $G$ is $k$-connected but $(k+1)$-disconnected,
there is a graph $H$ consisting of $k$ vertices such that $H$ separates $G$.  \\
{\bf 3)} For $G=K_n$, there is no subgraph $H$ which separates $G$. \\
{\bf 4)} The join of the $1$-sphere $C_4$ with the $0$-sphere $P_2$ is a 2-sphere, 
the disdyakis dodecahedron, a Catalan solid which by the way is $G \times K_1$, 
where $G$ is the octahedron. See Figure~(\ref{octahedron}). \\
The $C_4$ subgraph embedded as the equator in the octahedron $G$ 
separates $G$ into two wheel graphs $A,B$. \\

{\bf Remarks.} \\
{\bf 1)} A knot $H$ of co-dimension $2$ in a $3$-sphere $G$ is a closed simple 
curve embedded in $G$. Classical knots in $R^3$ can be realized in graph theory
as knots in $d$-spheres, so that the later embedding is the same as in the discrete version. 
The combinatorial problem is not quite equivalent however as it allows refined questions
like how many different knot types there are in a given $3$-sphere $G$ and how many topological 
invariants are needed in a given $3$-sphere to characterize any homotopy type or a knot
of a given co-dimension. . \\
{\bf 2)} An embedded curve has some ``smoothness". In a $2$-sphere for example, it intersects 
every triangle in maximally 2 edges. 
An extreme case is a Hamiltonian graph $H$ inside $G$ which by definition generates the 
entire graph. A Hamiltonian path $H$ which is a subgraph of 
a higher dimensional graph plays the role of a space-filling Peano curve in 
the continuum, a continuous surjective map from $[0,1]$
to the $2$-manifold $M$. For a simple curve $H$ which is not an embedding in $G$, the 
complement of $H$ can therefore be empty. 

\begin{defn}
A {\bf simple homotopy deformation} of a $(d-1)$-sphere $H$ in a $d$-sphere $G$ is
obtained by taking a $d$-simplex $x$ in $G$ which contains a non-empty set $Y$ of $(d-1)$-simplices 
of $H$ and replacing these simplices with $Y'$, the set of $(d-1)$-simplices in $x$ 
which are in the complement $Y$. 
\end{defn}

{\bf Examples.} \\
{\bf 1)} If $H$ is a simple curve in a $2$-sphere $G$ and $x$ is a triangle containing
a single edge $e$ of $H$, replace $e$ with the two other edges of the triangle.
This stretches the curve a bit. The reverse operation produces a ``shortcut" between two
vertices $(a,b)$ visited by the curve initially as $(a,c,b)$. \\
{\bf 2)} If $H$ is a $2$-dimensional graph, then a homotopy step is done by 
replacing a triangle in the tetrahedron with the 3 other triangles of a 
tetrahedron. \\

{\bf Remarks.} \\
{\bf 1)} It is allowed to  replace an entire $d$-simplex with the empty graph to 
allow a smaller dimensional sphere to be deformed to the empty graph.
This is not different in the continuum, where we deform curves to a point.  \\
{\bf 2)} A homotopy step does not honor embeddings in general. However, it preserves 
the class of simple curves with the empty curve included.  \\
{\bf 3)} We have in the past included a second homotopy deformation which removes or adds
backtracking parts $(a,b,a)$.
Since is only needed if one looks at homotopy deformations of general curves, we don't use it.
Homotopy groups must be dealt with using graph morphisms, rather than graphs. 
The backtracking deformation steps would throw us from the class
of simple curves. 

\begin{defn}
We say that a $k$-sphere $H$ is {\bf trivial} in a $d$-sphere $G$ if there is a 
sequence of simple homotopy deformations of $H$ which deforms $H$ to the empty graph. 
If every $1$-sphere is trivial in $G$, then $G$ is called {\bf simply connected}. 
If every $k$-sphere is trivial in $G$, we say the $k$'th homotopy class
is trivial.
\end{defn}

{\bf Remarks.} \\
{\bf 1)} The set of simple closed curves is not a group, as adding a curve to itself 
would cross the same point twice. Similarly, the set of simple $k$-spheres is not a group. 
In order to define the fundamental group, one has to look at graph homomorphism and not
at the images. This is completely analogue to the continuum, where one looks at 
continuous maps from $T^1$ to $G$. \\
{\bf 2)} Unlike in the continuum, 
where the zero'th homotopy set $\pi_0(G)$ is usually not provided with a group structure, 
but $\pi_0(G)$ can has a group structure. It is defined as the commutative 
group of subsets of $V$ with the symmetric difference $\Delta$ as addition, 
modulo the subgroup generated by sets 
$\{ \{a,b\} \; | \; (a,b) \in E\}$. It is of course $Z_0^{b_0}$ where $b_0={\rm dim}(H^0(G))$
is the number of connectivity components. \\
{\bf 3)} The Hurewicz homomorphism $\pi_0(G) \to H^0(G)$
maps a subset $A$ of $V$ to a locally constant function obtained by applying the heat flow
$e^{-L_0t}$ on the characteristic function $1_A(x)$ which is $1$ on $A$ and $0$ else, playing
the role of a $0$-current = generalized function in the continuum. \\
{\bf 4)} Also the Hurewicz homomorphism $\pi_1(G) \to H^1(G)$ is explicit by applying the heat
flow $e^{-L_1 t}$ on the function on edges telling how many times the curve has passed 
in a positive way through the edge. \\

{\bf Examples.} \\
{\bf 1)} Every simple curve in a $2$-sphere is trivial if it can be deformed to the empty graph. 
This general fact for $d$-spheres is easy to prove in the discrete setup because by definition,
a $d$-sphere becomes contractible after removing one vertex. The contraction of this punctured sphere
to a point allows a rather explicit deformation of the curve to the empty graph. \\
{\bf 2)} The deformation works also for $0$-spheres. In a connected graph, any embedded
$0$ sphere can be homotopically deformed to the empty graph. 
So, a graph is connected if and only if every $0$-sphere in $G$ is trivial.

\begin{defn}
Fix a geometric $d$-dimensional graph $G$.
Let $\pi_k(G)$ denote the union of all graph homomorphisms from a graph in the set 
$\{ C_k, k \geq 3 \}$ to $G$. Any such homomorphism $\phi: H \to G$ defines the 
{\bf homomorphism graph}, for which the vertices are the union of the 
vertices of $H$ and $G$ and for which the edge set is the union of the edges in $C$ and $G$ 
together with all pairs $(a,\phi(a))$. Two such homomorphisms are called homotopic, if the corresponding
homomorphism graphs are homotopic. The homotopy classes $\pi_1(G)$ define the {\bf fundamental group} of $G$. The
$0$-element in the group is the homotopy class of a map from the empty graph to $G$. 
The addition of two maps $C_k \to G, C_l \to G$ is a map from $C_{k+l} \to G$ obtained in the usual way 
by first deforming each map so that $\phi_i(0)=x_0$ is a fixed vertex $x_0$, then define 
$\phi(t)=\phi_1(t)$ for $t \leq k$ and then $\phi(t)=\phi_2(t-k)$ for 
$k \leq t \leq k+l$.
\end{defn}

{\bf Remarks.} \\
{\bf 1)} If one would realize the graph in an Euclidean space and see it as a triangularization 
of a manifold, then the fundamental groups of $G$ and $M$ were the same. 
The groups work also in higher dimensions. 
As we have to cut up a sphere at the equator to build the addition in the higher homotopy groups, 
it would actually be better to define the addition in the enhanced picture and look at maps $H_1 \to G_1$
where $H_1$ is the enhanced graph of the $k$-sphere $H$ and $G_1$ the enhanced version of the graph $G$. 
A deformation of a graph $H$ to a graph $K$ is then geometrically traced as a surface. \\
{\bf 2)} As a single basic homotopy extension step $C_k \to G$ to $C_{k+1} \to G$
keeps the map in the same group element of $\pi_1(G)$, the verification that the group operation is well 
defined is immediate.  \\

There are various generalized notions of ``geometric graphs", mirroring the definitions from the continuum. 
We mention them in the next definition, as we still explore discrete versions of questions related to 
Schoenflies problem in the continuum. The main question is whether there are discrete versions of 
exotic spheres, spheres which are homeomorphic to a $d$-sphere but for which unit spheres are not 
spheres. 

\begin{defn}
A {\bf homology $d$-sphere} is a geometric graph of dimension $d$ which has the same homology than a $d$-sphere. 
It is a geometric graph of dimension $d$ with Poincar\'e polynomial $p_G(x) = 1+x^d$.
A {\bf homology graph} of dimension $d$ is a graph for which every unit sphere is a homology sphere.
A {\bf pseudo geometric graph} of dimension $d$ is a graph for which every unit sphere is a finite
union of $(d-1)$ spheres. A {\bf discrete $d$-variety} is defined inductively as a graph for which
every unit sphere is a $(d-1)$-variety with the induction assumption that a $(-1)$-variety is the
empty graph.
\end{defn}

{\bf Examples.} \\
{\bf 1)} An example of a homology sphere can be obtained by triangulating the dodecahedron and doing identifications
as in the continuum. A suspension of a homology sphere is an example of a homology graph. \\
{\bf 2)} A figure eight graph is an example of a pseudo geometric graph of dimension $1$. \\
{\bf 3)} The cube graph or dodecahedron graph are examples of discrete $1$-varieties; their unit spheres are the
$0$-dimensional graphs $P_3$ which are not $0$-spheres but $0$-varieties.  

\begin{defn}
Two $k$-spheres $H,K$ in a $d$-sphere $G$ are called {\bf geometric homotopic within $G$}
if there is a geometric $(k+1)$-dimensional graph with boundary $M$ in the
$(d+1)$-dimensional graph $G \times L_n$ such that $M \cap (G \times \{0\}) = H \times \{0\}$ 
and $M \cap (G \times \{n\}) = K \times \{0\}$ and such that the boundary of $M$ 
is included in the boundary of $G \times L_n$. 
\end{defn}

{\bf Remarks.} \\
{\bf 1)} Given a $(d-1)$-sphere $H$ embedded in a $d$-sphere $G$. The deformation $H \to H'=H \Delta S(x)$
is equivalent to a homotopy deformation of the complement. \\
{\bf 2)} The above definition can can also be done for more general $k$-spheres (where $k$ is not
necessarily $d-1$) by taking intersections of unit spheres with $(k+1)$-sphere and performing the 
deformation within such a sphere.  \\
{\bf 3)} Any homotopy deformation of $H$ within $G$ defines a deformation of the embedding of
$H_1$ in $G_1$ and can be seen as a geometric homotopy deformation, a surface in $G \times L_2$. 
We will explore this elsewhere. 
Let $f$ be the polynomial in the variables $x_1,\dots,x_n$ representing vertices in $G$.
The function $f(x_1,\dots,x_n)(a + ab + b)$ describes the $(d+1)$-dimensional space
$G \times L_2$ in which we want to build a surface. Let $g(y_1,\dots,y_m)$ be the function describing the
surface $H$ in $G$, where $y_i$ are the vertices in $H$. Make the deformation at $x$: 
Define $a g(x_1,\dots,x_n) + a b [g_o(S(x)) x] + b [ x + g_e(S(x)) + g(x_1,\dots,x_n)]$. 
For example if $f_H = a+c+ac, f_K=a+b+c+ab+bc$, $f = u (a+b+c+a b+b c)+u v(a b c+a+c)+v(a+c+a c)$. \\
{\bf 4)} In \cite{KnillTopology} we wondered what the role of $1$ in the ring describing graphs could be. 
It could enter in reduced cohomology which is used in the
Alexander duality theorem $\overline{b}_k(H)=\overline{b}_{d-k-1}(G-H)$ (going back to \cite{Alexander15}). 
Here $b_k=\overline{b}_k$ if $k>0$ and $b_k=1+\overline{b}_k$ if $k=0$. 
In the Jordan case for example, where $b(H)=(1,1,0)$ and $b(G-H)=(2,0,0)$ one gets
$\overline{b}(H)=(0,1)$ and $b(G-H)=(1,0)$. When embedding a $2$ sphere $H$ in a $3$-sphere $G$,
then $b(H)=(1,0,1)$ and $\overline{b}(H)=(0,0,1)$ as well as  
$b(G-H)=(2,0,0)$ and $\overline{b}(G-H)=(1,0,0)$. This works for any $d$ as
$b(H)=(1,0,\dots,0,1)$ and $b(G-H)=(2,0,\dots,0)$ by Schoenflies. \\

\section{Tools}

In this section, we put together three results which will be essential in the proof. 
The first is the triviality of the fundamental group in a sphere:

\begin{lemma}[trivial fundamental group]
Every embedding of a $1$-sphere $H$ in a $d$-sphere $G$ for $d>1$ is 
homotopic to a point.
\end{lemma}
\begin{proof}
Look at the sphere $H_1$ in $G_1$.  Remove a vertex $x$ disjoint from $H_1$. 
By definition of a $d$-sphere, this produces a $d$-ball. The curve $H_1$ is contained in this ball $B$. 
We can now produce a homotopy deformations at the boundary of $B$ until $H$ is at the boundary. 
Note that $B$ does not remain a ball in general during this deformation as the boundary might not generate
itself but a larger set. But $B_1$ remains a ball. 
Once the $H$ is at the boundary switch an make homotopy deformations of $H$ until $H$ again
in the interior of $B$. Continuing like that, perform alternating homotopy deformations of $B$ and $H$. 
Because the ball $B$ can be deformed to a point, we can deform $H$ to the empty graph. 
\end{proof}

We now show that if sphere $H$ is a subgraph of $G$, then the enhanced sphere $H_1$ is 
embedded. It is an important point but readily follows from the definitions: 

\begin{propo}
If $H$ is a $(d-1)$-sphere which is a subgraph of a $d$-sphere $G$, then $H_1$ is
embedded in $G_1$. 
\end{propo}
\begin{proof}
We have to show that $H_1 \cap S(x_1) \cap \dots \cap  S(x_k)$ is a $(d-k-1)$-sphere 
for every $k$ and every simplex $x$ with vertices $x_1,\dots,x_k$ in $H_1$. 
Any intersections are the same whether we see $H$ as part of $G$ or
whether $H$ is taken alone. The reason is that any of the unit spheres $S(x_k)$ consists
of simplices which either contain $x_k$ or are contained in $x_k$. None of these
simplices invoke anything from $G$. 
So, the statement reduces to the fact that the intersection of spheres $S(x_k)$ with 
$x_k$ belonging to a simplex form a sphere. But this is true by induction. For one
sphere it is the definition of a sphere. If we add an additional sphere, we drill
down to a unit sphere in a lower dimensional sphere. 
\end{proof}

Finally, we have to look at an intersection number. 
At appears at first that we need a transversality condition when describing
spheres $K,H$ of complementary dimension $1,d-1$ in a $d$-sphere $G$. 
While we will not require transversality, the notion helps to visualize the situation. 

\begin{defn}
Given an embedding of a $(d-1)$-sphere $H$ in a $d$-sphere $G$ and a simple curve $C$
in $G$. We say it {\bf $C$ crosses $H$ transversely} if for every $t$ such that $C(t) \in H$,
both $C(t-1)$ and $C(t+1)$ are not in $H$.
\end{defn}

More generally: 

\begin{defn}
Given a $d$-sphere $G$, let $K$ be an embedded $k$-sphere and let $H$
be an embedded $(d-k)$ sphere. We say $K,H$ {\bf intersect transversely} if 
$H_1,K_1$ intersect in a $0$-dimensional geometric graph. 
\end{defn} 

{\bf Remarks} \\
{\bf 1)} Given two complementary spheres $H,K$ in a sphere $G$. Look at the spheres $H_1,K_1$
in $G_1$. There is always a modification of the spheres so that they are transversal. 
Consider for example the extreme example of two identical $1$-spheres $H,K$ in the equator of 
the octahedron $G$. The graphs $H_1,K_1$ are closed curves of length $8$ inside the
Catalan solid $G_1$. Now modify the closed curves by forcing $H_1$ to visit the 
vertices in $G_1$ corresponding to the original $4$ upper triangles and $K_1$ visit the 
vertices in $G_1$ corresponding to the original $4$ lower triangles. The modified
curves now intersect in $4$ vertices. \\
{\bf 2)} Given a $(d-1)$-sphere $H$ embedded in a $d$-sphere $G$ and given a closed curve $C$
which is transverse to $H$ let $x_j$ be the finite intersection points. We need to 
count these intersection points. For example, if a curve is just tangent to 
a sphere, we have only one intersection point even so we should count it with
multiplicity $2$. When doing homotopy deformations, we will have such situations
most of the time. As we can not avoid losing transversality when doing deformations,
it is better to assign intersection numbers in full generality, also if we have no
transversality. \\

For the following definition, we fix an orientation of 
the $1$-sphere $C$ which we only require to be a simple closed curve and not an
embedding and we fix also an orientation on the $(d-1)$-sphere $H$. 
Since both are spheres and so orientable, this is possible. 

\begin{defn}
If the vertex set of $C$ is contained in the vertex set of $H$, we 
define the {\bf intersection number} to be zero. Otherwise,
let $\{ a=C(t_1),\dots,C(t_k)=b \} \subset H_1$ be a connected time interval in $H_1$. 
Let $y=C(t_0)$ be the vertex in $G \setminus H$ just 
before hitting $H_1$ and $z=C(t_{k+1})$ the vertex just after leaving $H_1$. 
The vertex $a$ is contained in a $d$-dimensional simplex generated by the edge $(y,a)$ in $C$ and a 
simplex $\sigma$ in $H_1$ containing $a$. This $d$-simplex has an orientation from the orientation
of the edge in the curve and the simplex $\sigma$ in $H$.
If the orientation of the simplex generated by $C$ and $H$ at the point agrees with the orientation 
of the $d$-simplex coming from $G_1$, 
define the {\bf incoming intersection number} to be $1$. Otherwise define it to be $0$. Similarly, we have an 
outgoing intersection number which is again either $0$ or $1$. 
The sum of all the incoming and outgoing intersection numbers is called the {\bf intersection number} of the curve. 
\end{defn}

In other words, if the curve $C$ touches $H$ and bounces back after possibly staying in $H$ for some time 
then the intersection number is $2$ or $-2$, depending on the match of orientation of this ``touch down".
If the curve $C$ passes through $H$, entering on and leaving on different sides, then 
the intersection number belonging to this ``crossing" is assumed to be $1$. 
If a curve is contained in $H$, then the intersection number is $0$. 
The intersection number depends on the choice of orientations chosen on $C$ and $H$, 
but these orientations only affect the sign. 
If the orientation of the curve or the $(d-1)$-sphere is changed, the sign changes, but only by an 
even number. 

\begin{lemma}
Two homotopic curves in a sphere $G$ with embedded $(d-1)$ sphere 
have the same intersection number modulo $2$. 
\end{lemma}
\begin{proof}
Just check it for a single homotopy deformation step of the curve. As these steps are local,
a change of the intersection number can only happen if a curve makes a touch down at $H$
before the deformation and afterwards does no more intersect. This changes the intersection
number by $2$. 
\end{proof}

The following lemma will be used in Jordan-Brouwer and is essentially 
equivalent to Jordan-Brouwer:

\begin{lemma}
\label{keylemma}
The intersection number of a closed curve $K$ with an embedded  $(d-1)$-sphere $H$ in a
$d$-sphere $G$ is always even.
\end{lemma}

\begin{proof}
The sphere $G$ is simply connected, so that every curve $C$ can be deformed to a vertex
not in $H$. For the later, the intersection number is zero. 
\end{proof}

In the case of transversal intersections, the result can be illustrated geometrically. 
Look at the intersection of the $2$-dimensional deformation surface of the curve
with the $d$-dimensional deformation cylinder $H \times L_n$. 
As long as the curve never intersects $H$ in more than two adjacent vertices,
this is a $1$-dimensional geometric graph, which must consist of a finite union 
of circular graphs. 

\section{The theorems}

\begin{thm}[Discrete Jordan-Brouwer]
A $(d-1)$-sphere $H$ embedded in a $d$-sphere $G$ separates $G$
into two complementary components $A,B$. 
\end{thm}

\begin{proof}
We use induction with respect to the dimension $d$ of the graph $G$. 
For $d=0$, the graph $G$ is the $0$-sphere and $H$ is the empty graph $S_{-1}$ and the 
complementary components $A,B$ are both $K_1$ graphs which by definition are $0$-balls.  \\

To prove the theorem, we prove a stronger statement: if $H$ is a $(d-1)$-sphere which 
is a subgraph of $H$, then $H_1$ separates $G_1$ into two components $A,B$.  \\

Take a vertex $x \in H \subset G$. When removing it, by definition of spheres, it 
produces the $(d-1)$-ball $H' = \setminus \{x\}$ and the $d$-ball $G' = G \setminus \{x\}$.
By definition of embedding, 
the boundary of $H'$ is $(d-2)$-sphere which is a subgraph of the boundary $S(x)$ of $G'$. 

By induction, $H'_1$ divides $G'_1$ into two connected parts $A'$ and $B'$. 

Let $A''$ be the path connected component in $G_1 \setminus H_1$ containing $A'$ and let $B''$ be the path 
connected component in $G_1 \setminus H_1$ containing $B'$. The union of $A''$ and $B''$ is $G$
if there was a third component it would have $H$ in its boundary and
so intersect $S(x)$, where by induction only two components exist.

The graphs $A=A'' \cup H_1$ and $B=B'' \cup H_1$ 
are $d$-dimensional graphs which cover $G_1$ and have the $(d-1)$ sphere $H_1$ as a common boundary. \\

It remains to show that the two components $A'',B''$ are disconnected. 
Assume they are not. Then there is a path $C$ in $G_1$ connecting a vertex 
$x \in A' \subset S(x)_1$ with a vertex $y \in B' \subset S(x)_1$ such that $C$ does not 
pass through $H$.  Define an other path $C'$ from $x$ to $y$ but within $S(x)$. 
The path $C'$ passes through $H$ and the sum of the two paths $C,C'$ form a closed path
whose intersection number with $H$ is $1$. 
This contradicts the intersection lemma~(\ref{keylemma}).
\end{proof}

{\bf Remarks.} \\
{\bf 1)} As we know $\chi(G)=1+(-1)^d$ and $\chi(H)=1-(-1)^d$ 
we have $2=\chi(G)+\chi(H) = \chi(A) + \chi(B)$ so that $\chi(A) + \chi(B) = 2$. 
We will of course know that $\chi(A)=\chi(B)=1$ but $\chi(A) + \chi(B)=2$ is for free. \\
{\bf 2)} The proof used only that $G$ is simply connected. We don't yet know that $A,B$ 
are both simply connected. At this stage, there would still be some possibility that one
of them is not similarly as the Alexander horned sphere. \\

{\bf Examples.}  \\
{\bf 1)} The empty graph separates the two point graph $P_2$ without edges 
into two one point graphs. 
{\bf 2)} The two point graph $P_2$ without edges separates $C_4$ into two $1$-balls $K_2$. \\
{\bf 3)} The equator $C_4$ in an octahedron $O$ separates $O$ into two wheel graphs $W_4$. \\
{\bf 4)} The unit circle $S(x)=C_5$ of a vertex $x$ in an icosahedron separates it into
a wheel graph $W_5$ and a $2$-ball with 11 vertices. \\
{\bf 5)} The $(d-1)$-dimensional cross polytop embedded in a $d$-dimensional cross
polytop separates it into two $d$-balls.  \\
{\bf 6)} More generally, a $d$-sphere $H$ divides the suspension $G=S_0 \star H$ 
into two balls, which are both suspensions of $H$ with a single point. \\

\begin{thm}[Jordan-Schoenflies]
A $(d-1)$-sphere $H$ embedded in a $d$-sphere $G$ separates $G$
into two complementary components $A,B$ such that
$A,B$ are both $d$-balls. 
\end{thm}

\begin{proof} 
Jordan-Brouwer gives two complementary domains $A,B$ in $G$. 
They define complementary domains $A_1,B_1$ in $G_1$, where $A_1,B_1$ have a nonzero number of 
vertices. The proof goes in two steps. We first show that $A_1,B_1$ are balls:
to show that $A_1$ is a $d$-ball, (the case of $B_1$ is analog), 
we only have to show that $A_1$ is contractible, as the boundary is by definition the sphere $H$.
It is enough to verify that we can make a homotopy deformation step of $H$ so that 
the number of vertices in $A_1$ gets smaller. This can be 
done even if the complement of $H$ has no vertices left, like for example if $H$ would be a Hamiltonian path
where $A$ is empty. A simple homotopy deformation of $H$ will reduce the number of vertices 
in $A_1$ and since $A_1$ only has a finite number of vertices, this lead to a situation, where $H_1$
has been reduced to a unit sphere meaning that $H$ has reduced to a simplex verifying that 
$A_1$ is a $d$-ball. The deformation step is done by taking any vertex $x$ in $A_1$ which belongs to a 
$d$-simplex $x$ in $A$. The step $H \to H \Delta x$ produces the step $H_1 \to H_1 \Delta S(x)$ in $G_1$. 
This finishes the verification that $A_1,B_1$ are $d$-balls. But now we go back to the 
original situation and note that $A$ and $A_1$ originally are homeomorphic. There is a
concrete cover of $A_1$ with $d$-balls such that the nerve graph is $A$. So, also 
$A$ is contractible. 
\end{proof} 

{\bf Remarks. } \\
{\bf 1)} Without the enhanced picture of the embedding $H_1$ in $G_1$,
there would be a difficulty as we would have to require the embedded graph $H$ to remain 
an embedding.  The enhanced picture allows to watch the progress of the deformation.  \\
{\bf 2)} As shown in the proof, one could reformulate the result by saying that 
if a $(d-1)$-sphere $H$ is a subgraph of a $d$-sphere $G$ (not necessarily an embedding)
then $H_1$ separates $G_1$ into two complementary components $A_1,B_1$ 
which are both $d$-balls.

\bibliographystyle{plain}

\end{document}